\documentclass[onecolumn,submit]{scias} 

\FrontMatterExtraVskip{0.0pt} 

\usepackage{xcolor}
\usepackage{graphicx}
\usepackage{subcaption}
\usepackage{tikz}
\usepackage{tabularray}
\usepackage{algorithm}
\usepackage{algpseudocode}

\title{An Algorithm for Illuminating $n$ Nonoverlapping Circular Discs' Boundaries on the Plane with Application to Tree Stem Illumination Problem}

\author[a]{Phapaengmuang Sukkasem} 
\author[b,*]{Supanut Chaidee} 
\author[c,d]{Watit Khokthong}                                                           %

\address[a]{Program in Applied Mathematics, Department of Mathematics, Faculty of Science, Chiang Mai University, Chiang Mai 50200, Thailand}
\address[b]{Advanced Research Center for Computational Simulation (ARCCoS), Department of Mathematics, Faculty of Science, Chiang Mai University, Thailand}
\address[c]{Forest Restoration Research Unit, Department of Biology, Faculty of Science, Chiang Mai University, Chiang Mai 50200, Thailand}
\address[d]{Environmental Science Research Centre, Faculty of Science, Chiang Mai University, Chiang Mai 50200, Thailand}

\ead{supanut.c@cmu.ac.th} 

\abstract{Given a set of $n$ nonoverlapping circular discs on a plane, we aim to determine possible positions of points (referred to as cameras) that could fully illuminate all the circular discs' boundaries. This work presents a geometric approach for determining feasible camera positions that would provide total illumination of all circular discs. The Laguerre Delaunay triangulation, coupled with the intersection of slabs formed by the boundaries of circular discs, is employed to form the region that satisfies the given conditions. The experiment is conducted using a set of randomly positioned circular discs on a plane. This study has the potential to address the issue of illumination in forests by utilizing a LiDAR camera to determine the possible number and placement of cameras that can effectively illuminate trees within a forest.}



\begin{document}
\frontmatter

\section{Introduction}

In our everyday observation, as well as in mathematical views as rooms or open areas with obstacles, we have the challenge of identifying the regions to circumvent obstructions in the visual field. From a mathematical perspective, the illumination problem is well established in discrete and computational geometry \cite{1startgal}. The main issue of illumination problems \cite{illuprob} is to find the least number of lighting points to obtain the highest illumination efficiency \cite{artgal92} and to place the lighting source so that it can illuminate the area or object \cite{artgal92}. There are widely known problems in the illumination problem, such as the Art Gallery problem \cite{artgalprob} \cite{artgal92} \cite{artgalalg}, the Guarding problem \cite{terrainguard} \cite{locateguard}, and the floodlight illumination problem \cite{illurectri} \cite{floodlightref1}. In various conditions for illumination, the lighting source that is positioned can be projected to test the efficiency. 

A practical example of the challenge of illumination is the use of LiDAR (Light Detection and Ranging) sensors to illuminate forest structures for environmental research \cite{forestcanoillu}. Forest practitioners have utilized the LiDAR \cite{lidarorigin} and camera sensors by the structure from motion to illuminate the forests. Focusing on the LiDAR, regarded as a light source capable of emitting illumination in all directions, the challenge lies in identifying the optimal position for the LiDAR laser scanner to maximize the capture of forest structural information. Recent studies focused on the suggestion of the use of terrestrial laser scanner (TLS) with forest. A study in Abegg et al. \cite{tls}, investigated the impact of different patterns of placement of scanners on visibility within forest stands when using terrestrial laser scanning. By simulating various scanning patterns and positions, the study demonstrated that the arrangement and number of scanning locations significantly influence mean visibility.  Another study by Wilkes et al. \cite{tls2} discussed the use of TLS to capture detailed forest structure data on extensive plots. The study recommended a systematic grid pattern for the scan position locations, highlighting the effectiveness of a 10m×10m grid in achieving uniform point distribution. 

The challenge of non-detection in single scan TLS of forests is addressed in a study from 2019 \cite{tls3}, which can miss some tree stems or other vegetative parts that are obscured by others. It introduced a new estimator in terms of marked point pattern of trees that compensates for these non-detection results, by considering different detection rules, which model partial or nearly full or visibility of a standing tree. The study in 2020 \cite{tls4} proposed an Iterative-Mode Scan Design for TLS in forests to minimize occlusion effects and improve accurate tree attribute estimation. The method evaluates the potential location for the scan using a PoTo index and the cumulative degree of ring closure (CDRC) to enhance the completeness of the scan, particularly in dense stands. 

Although the previously mentioned literature has suggested basic strategies for gathering tree-root information from TLS LiDAR, the specific procedure for determining the optimal placement of the active sensor remains improperly explored. To model this problem mathematically, we can represent the cross-section of the trees as a set of circles, with the TLS LiDAR serving as the emission light source. We aim to identify the optimal positions for TLS LiDAR to effectively illuminate a specified set of circles that can represent tree stems and to minimize the number of sources while maximizing the illuminated area. 

The study of Tóth \cite{illuconvexdisc} explores the illumination of convex discs and explains how a set of $n$ circular discs can be illuminated by $2n$ points. Then there is more study of scenarios where circles overlap or are centered on a line. The concept of a generalized Voronoi diagram is used to describe how vertices can illuminate the corresponding arcs of a circle. Theorems are presented to establish upper bounds for the required number of lighting points, affirming the sufficiency of $2n$ points, and suggesting $\max(2n,4n-7)$ points as necessary for nonoverlapping open convex discs. Those propose more results and combinatorial aspects in the art gallery problem in point guard and edge guard \cite{artgal92} \cite{illuconvex}.

In this paper, we present a geometrical approach to ascertain the locations of light sources required for the total illumination of a collection of circles, guaranteeing that each circle is entirely illuminated by the sources. A slab characterized by a region bounded by parallel lines formed by two adjacent circles will be investigated to establish the feasible region for the placement of light sources. We prove the geometrical properties of intersecting slabs, then utilize these areas as potential locations for source placement, and propose an algorithm that addresses the problem. Subsequently, we conduct experiments to determine the optimal scenario using ideally generated data comprising both identical and varied circles. We also use artificial trees and utilize LiDAR to illuminate these trees, so obtaining the illuminated region serves as validation from the computations to quantify the accuracy from the ideally generated data. 

This paper is organized as follows. The introduction includes a related review on the research topic. In the second section, we provide the basic background in computational geometry, which we have used in this study. In addition, we formulate the problem to tackle in this part. The geometrical properties of the illuminating region will be investigated in Section 3, including an algorithm to locate the illuminating location. In Section 4, we perform experiments to verify the proposed algorithm. We finally conclude our study in the last section of this paper.

\section{Preliminaries}

\subsection{Background in Computational Geometry}

The \textbf{convex hull} \cite{compgem} of a set $S$, denoted $CH(S)$, is the smallest convex set that contains $S$. It is the intersection of all convex sets that contain $S$, i.e., for any $\alpha \in \Lambda$ and $C_\alpha$ a convex set covering $S$, $CH(S)= \bigcap\limits_{\alpha \in \Lambda} C_\alpha$. There are several algorithms to compute the convex hull from a set of finite points on the plane, such as divide and conquer. The time complexity of the convex hull algorithm is known to be $O(n\log n)$.

A set of points $p_i \in P$ where $i=1,2,\dots,n$ on a convex hull segment of a convex set $S$ is \textbf{visible} \cite{vis} to a light point $q$ if and only if there exists a non-crossing line segment $\overrightarrow{vp_i}$ between all points in $P$ to $q$. The \textbf{illumination problem} is a classic problem to find the location of the light source $s$ that an object $O$ is visible to $s$. We call $s$ an illuminating point and $O$ is illuminated by $s$.  It is remarked that the boundary of a convex set in the plane can always be illuminated by using three light points \cite{illuprob}. 

One of the most used diagrams in computational geometry is the Voronoi diagram, as compiled in the book \cite{Okabe}. Let $P=\left\{ p_i,\dots,p_n\right\}$ be a set of points on a plane called a set of generator points, and let $d_E(p,p_i)$ denote the (Euclidean) distance between the points $p$ and $p_i$ on the plane. The Voronoi region $V(p)$ of a site $p$ in $S$ is $$V(p) =\left\{p \in \mathbb{R}^2 | \norm{p-p_i} \leq \norm{p-p_j} \text{ where } i,j=1,\dots,n \text{ and } i\neq j\right\}$$ such that the Voronoi diagram \cite{compgem} of a generator set $P$ is the Voronoi region written as $Vor(P)=\left\{V(p_i),\dots,V(p_n)\right\}$. The \textbf{Voronoi diagram} corresponds to the Voronoi edges and Voronoi vertices that divide the plane into regions. The cells of the partitioning of a plane are called Voronoi polygons. The \textbf{Delaunay triangulation} \cite{compgem} is generated by the line segment between $p_i$ and $p_j$ if and only if the boundaries of $V(p_i)$ and $V(p_j)$  share a common edge.

Based on Tóth \cite{illuconvexdisc}, which discusses the use of the vertex of the Laguerre Voronoi diagram as the placement of light points, the definition of the Laguerre Voronoi diagram and the Laguerre Delaunay triangulation are necessary to declare, which is the duality of each other.

Given $G=\left\{C_i,\dots,C_n\right\}$ be a set of $n$ circular discs in the plane, let $r_i$ and $p_i=(x_i,y_i)$ be the radius and the center of the circular disc $C_i$. For a disc $C_i$ and point $p=(x,y)$ in the plane, let us define $d_L(p,C_i)$ by $$d_L(p,C_i)=d(p,p_i)^2-r_i^2.$$ Let us define $$V_L(G,C_i)=\left\{p\in \mathbb{R}^2|d_L(p,C_i)<d_L(p,C_j),i\neq j\right\}.$$ The plane is partitioned into $V_L(G,C_i),\dots,V_L(G,C_n)$ and their boundaries. The partition is called the \textbf{Laguerre Voronoi diagram} \cite{Aurenhammer, Imai} or the power diagram.

For a given set $G$, the \textbf{Laguerre Delaunay diagram }\cite{laguerre} is generated by the following procedure. We draw a line segment between $p_i$ and $p_j$ if and only if the boundaries of $V_L(G,C_i)$ and $V_L(G,C_j)$  share a common edge. For circular disc $C_i$, we define $C_i^*$ as the 3D point mapped by $(x_i,y_i)\mapsto x_i^2+y_i^2-r_i^2$; i.e., defined by $(x_i, y_i, x_i^2+y_i^2-r_i^2)$ for $i=1,\dots,n$, and define $G^*=\left\{C_i^*,\dots,C_n^*\right\}$. The projection of the lower 3D convex hull of $CH(G^*)$ onto the $xy$ plane is the Laguerre Delaunay diagram for $G^*$.

In this study, we focus on considering regions induced by tangent lines. In general, a \textbf{slab} is the region bounded by two parallel lines $l_i$ and $l_j$.

\subsection{Problem formulation}

Given the forest with $n$ trees in three-dimensional space, we are going to project the cross section of the forest into a set of $n$ circular discs in two-dimensional space. The problem of three-dimensional forest scanning can be addressed with the illumination problem of positioning the illuminating point that covers most of the boundaries of the $n$ circular discs, given the LiDAR and the camera sensor as the illuminating point to illuminate the boundaries of the forest as a set of $n$ circular discs. 
In this study, our objective is to construct a deterministic geometrical algorithm to find a set of appropriate regions to position the illuminating point.

\section{Geometrical Properties of Illuminating Region}
\subsection{Properties of slabs}

From the previous section, we are going to represent the forest as a set of circular discs and the LiDAR and camera sensors as illuminating points. Therefore, we analyze the properties of the halfplane that is constructed by a tangent line corresponding to a point on the circumference of the circular disc to discover the properties of an illuminating point.

\begin{lemma}
	Let $C$ be a circular disc, $p$ be a point on the circular disc $C$, and $l_p$ be the tangent line to $C$ at $p$. If the illuminating point $q$ is on the side of halfplane corresponding to $l_p$ which does not contain $C$, then $p$ is always illuminated by $q$. \Llabel{illuminate}
\end{lemma}

\begin{proof}
	At the illuminating point $q$, we can construct two tangent rays to the circular disc $C$ without obstacle as shown in Figure \ref{lemanddef} (a). The line passing through two tangent points $q_i$ and $q_j$ then separates the circular discs into two arcs, $a_i$, which are aimed at the point $q$ and $a_j$ opposite to $a_i$. If point $p$ lies on $a_i$ then there exists a line segment $pq$ that all parts do not lie inside $C$. On the other hand, if the point $p$ is in $a_j$ then some part of $pq$ lies inside $C$. Thus, $q$ only illuminates $p$ that lies in the arc $a_i$. Therefore, we consider $p$ that lies in $a_i$ then construct the tangent line $l_p$ to $C$ in $p$. We have that $q$ is always on a side of halfplane of $l_p$that does not contain $C$ since tangent line of any end point of $a_i$ contain ray corresponding to $q$ and that end point.
\end{proof}

\begin{figure}[h!t]
	\centering
	\includegraphics[scale=0.3]{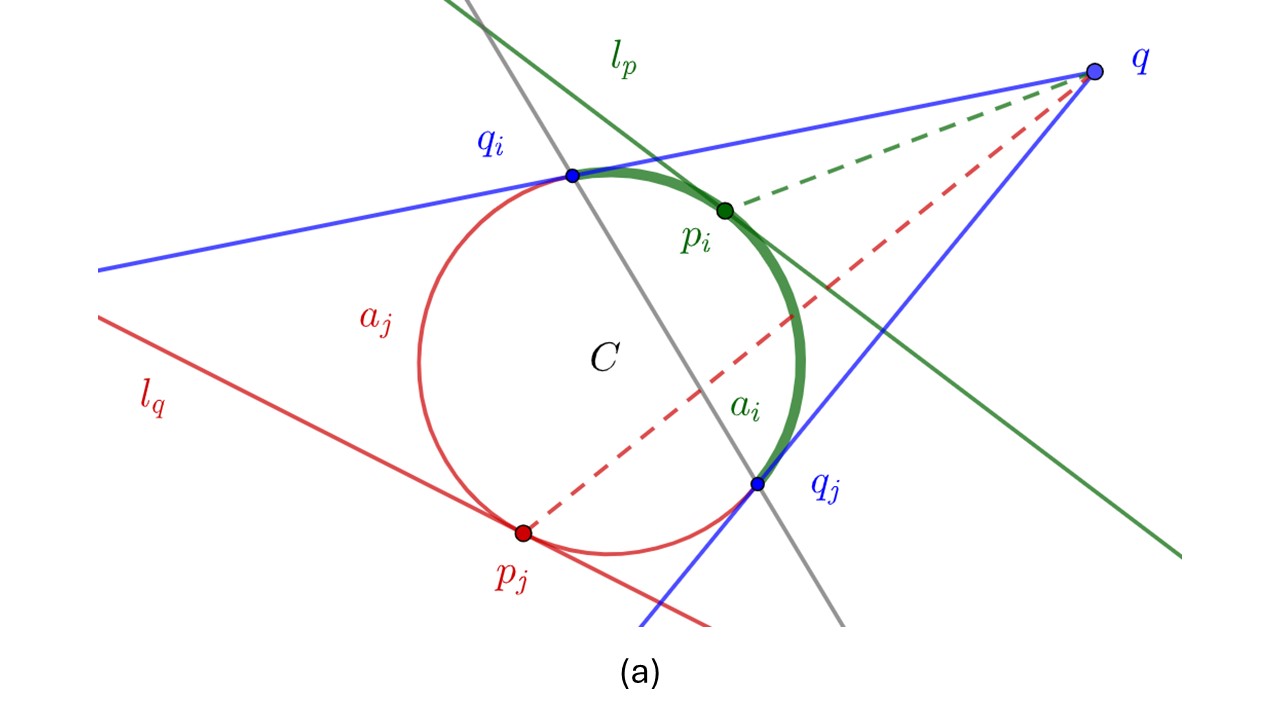}
	\includegraphics[scale=0.3]{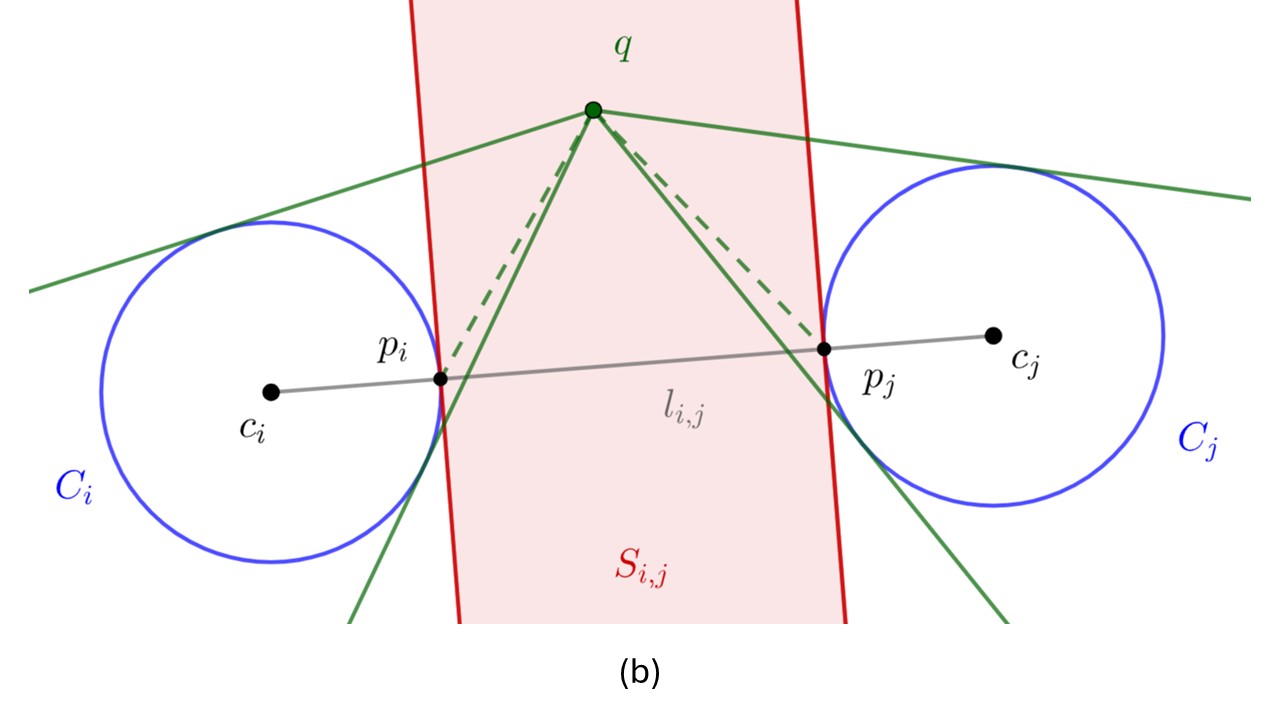}
	\caption{(a) a point $p_i$ on arc $a_i$ is visible to $q$ but a point $p_j$ on another arc is not visible and (b) the slab $S_{i,j}$ of circular discs $C_i$ and $C_j$ with point $p_i$ and $p_j$ visible to illuminating point $q$ on $S_{i,j}$}
	\label{lemanddef}
\end{figure}

\newpage

From \Lref{illuminate}, we extend the idea to two circular discs with an illuminating point. We construct two parallel tangent lines on each circular disc that correspond to the intersection point between their circumference and line segment from center of both circular discs. We define a slab as the region of intersection of two halfplanes that do not contain any circular discs from parallel tangent lines will be a region for illuminating point location, by the following definition.

\begin{definition} Let $C_i$ and $C_j$ be two nonoverlapping circular discs and $l_{i,j}$ be a segment joining the center $c_i$ and $c_j$ of these two circular discs. The region $S_{i,j}$ bounded by pair of parallel tangent lines $l_i$ and $l_j$, which are perpendicular to $l_{i,j}$ at the intersection points $p_i$ and $p_j$ between the segment and circular discs boundaries, of these two circular discs is called slab.
\end{definition}

Remark that the slab is unique up to the given size of circular discs and the distance between two circular discs. Therefore, we would verify the visible property of an illuminating point in a slab by the following lemma.

\begin{lemma}
	Any point on a slab always illuminates two intersection points on different two nonoverlapping circular discs which correspond to boundaries of the slab. \Llabel{slabillumination}
\end{lemma}

\begin{proof}
	Given two nonoverlapping circular discs $C_i$ and $C_j$, there exists a slab constructed by intersection points $p_i$ and $p_j$ of the line segment $l_{i,j}$ between the center of two circular discs and their circumference. The half-planes corresponding to these two intersection points is the boundary of the slab. If the illuminating point $q$ is on the side of both halfplanes which contain $C_i$, by \Lref{illuminate} then $q$ cannot illuminate the point $p_i$. This also implies for $q$ on the side of the half-plane with $C_j$. Thus, the illuminating point $q$ must be in the region that does not include circular discs, which is the area of the slab, to illuminate $p_i$ and $p_j$ simultaneously.
\end{proof}

For a collection of $n$ circular discs, we consider neighboring circular discs using Laguerre Delaunay triangulation of those circular discs. Hence, for any non-degenerate triangulation, we construct a slab of each two circular discs from three circular discs corresponding to a triangle in the triangulation.

We apply the idea of intersection points on two circular discs into three circular discs that can construct circular arcs on each circular disc that have these intersection points as their end point. In fact, any convex body uses at least three illuminating points to illuminate all its boundaries \cite{illuprob}, and half of the convex body cannot be illuminated by one illuminating point. For any three nonoverlapping circular discs, circular arcs that have intersection points as end points and arc lengths that are less than half the boundary size are called \textbf{objective arcs}.

\subsection{Intersection of slabs}
In previous section, we see that the slab is a region between parallel lines to position the illuminating point of two circular discs. In this section, we are going to observe the intersection of slabs, which will be leads to a feasible region of illuminating point position, and explore the geometric shape and properties of feasible region constructed by the intersection of three slabs generated by three circular disc. 

Remark that two groups of lines that compose the slabs have different slopes. Then we observe that the region constructed by the intersection of two slabs is a parallelogram as shown in Figure \ref{parallelogram}

\begin{figure}[h!t]
	\centering
	\includegraphics[scale=0.56]{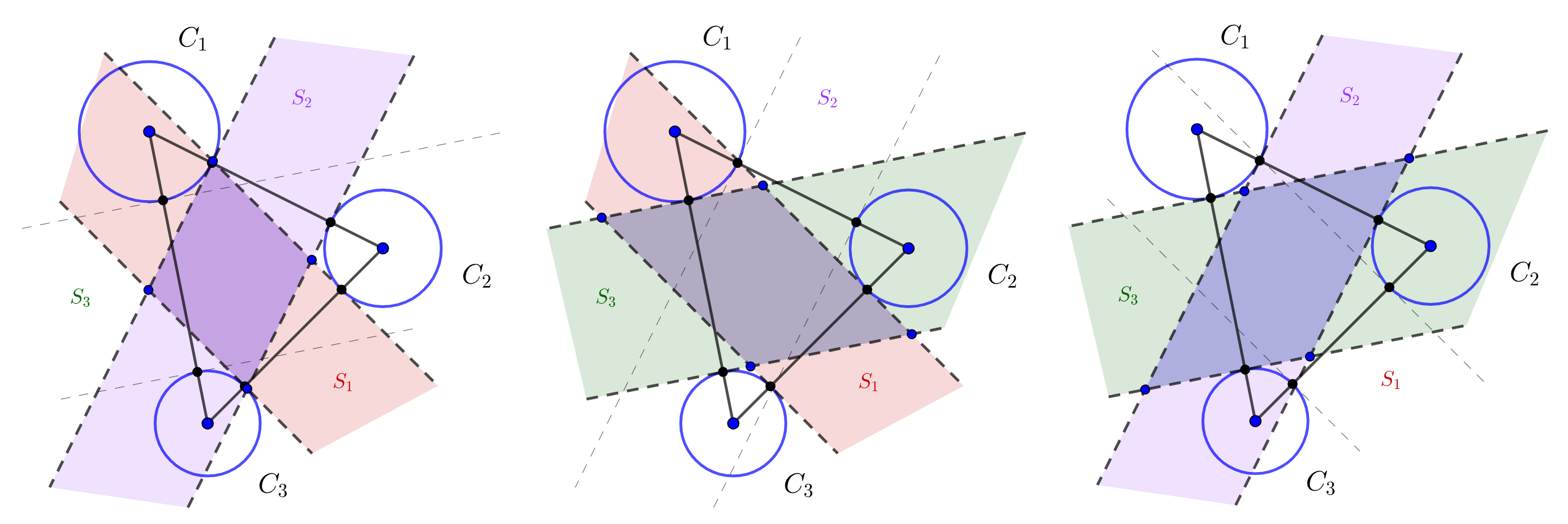}
	\caption{Every pair of two slabs from constructed by three nonoverlapping circular discs forms a parallelogram.}
	\label{parallelogram}
\end{figure}

\begin{definition} Let $C_i$, $C_j$ and $C_k$ be three nonoverlapping circular discs. The region constructed by intersection of three slabs construct by these three circulars is called \textbf{feasible region}.
\end{definition}

\begin{lemma}
	The feasible region of three nonoverlapping circular discs is a convex polygon. \Llabel{feasibleconvex}
\end{lemma}

\begin{proof}
	Given three nonoverlapping circular discs $C_i$, $C_j$, and $C_k$, we can construct three slabs $S_{i,k}, S_{j,k}$ and $S_{i,k}$, which correspond to each pair of circular discs. Since each pair of slabs generates a parallelogram, we have three parallelograms $P_{i,k}, P_{j,k}$ and $P_{i,k}$ from the intersection of every pair of three slabs. The parallelogram is a convex polygon that is the boundary of a convex set. In fact, since the intersection of convex sets is convex, the intersection of parallelograms $P_{i,k}, P_{j,k}$ and $P_{i,k}$ constructs a convex region. Thus, the feasible region is a convex polygon.
\end{proof}

\begin{theorem}
	Any point in a feasible region of three nonoverlapping circular discs always illuminates all objective arcs. \Tlabel{feasible}
\end{theorem}

\begin{proof}
	Given three nonoverlapping circular discs, for each triangle in the Laguerre-Delaunay triangulation, there are 6 intersection points between triangle edges and circular disc boundaries, which are two on each circular disc. Any two points on each circular disc form two circular arcs, which are objective arc and non-objective arc. From \Lref{slabillumination}, the illuminating point can illuminate two intersection points on different circular discs at the same time when it lies on the slab. Since the feasible region is the intersection of three slabs, any point on the feasible region always illuminates all points between two end points of the objective arc on each circular disc. Thus, the point on feasible region always illuminates all three objective arcs. 
\end{proof}

From \Tref{feasible}, we consider the objective arcs to be the region which is generated to be illuminated. Since the illuminating point is in the feasible regions, we ensure that the objective arc is the least arc of circular disc boundary illuminated by illuminating point.

\subsection{Shapes of feasible region}
We consider the possible shape of a feasible region with respect to the set of circular discs by the following theorem.

\begin{theorem}
	The possible shape of the feasible region ranges from 4-gon to 6-gon
\end{theorem}

\begin{proof}
	We first prove that the feasible region is not a triangle. Assume that $T$ is the feasible region in triangular shape with corresponding triangle $T_{\triangle}$ in the Laguerre Delaunay triangulation, i.e. it corresponds to the circles $C_i, C_j$ and $C_k$, then we consider the intersection points between the boundaries of $C_i, C_j, C_k$ and $T_{\triangle}$ that generate slabs $S_{i,j}, S_{i,k}$ and $S_{j,k}$ of each pair of circles. Then there exists a circle that touches two slabs, so that the vertex of intersection of two slabs is a vertex of the feasible region. Without loss of generality, assume that $C_i$ and that each circle of $C_j$ and $C_k$ has at least one intersection, which is on the same triangle edge joining to $C_i$. Note that the intersection point would be on the same line passing the feasible regions vertices. If $C_i$ is constructed from a feasible region in triangular shape, then $C_i$ would touch the line of two slabs $l$ and $m$ with the same distance to the vertex $v$ of the feasible region generated by these two lines called $d$. Therefore, we construct another line $k$ to $C_i$ which is perpendicular to the lines $l$ and $l^\prime$ of the slab. Then $k$ intersects $l$ and $l^\prime$ at points $p$ and $p^\prime$. $p$ is on the same line as $l$, but $q$ is not on the same line of any line that makes up the feasible region. This gives us a contradiction, as illustrated in Figure \ref{shape} (a).

    For special case where the line of slab concurrently is a triangle, if $C_i$ is constructed as the same as previously mentioned by a line $l$ and $m$, then there is a line $l, m$ and $n$ as the line of slabs, which intersects each other. $C_j$ is constructed by line $l^\prime$ and $n^\prime$. Construct a tangent line $L$ of $C_i$ and $C_j$ in such a way that $C_i$ and $C_j$ are on the same side and $L$ crosses the feasible region, then $C_k$ would be on the opposite side of $C_i$ and $C_j$ but the constructed line force $C_k$ to be on the same side of $C_i$ and $C_j$.

    Secondly, it is possible to be a parallelogram (4-gon) in the case that the first two slabs intersect to be a parallelogram, and the parallel lines of the third slab concurrent the intersected points, or cover the existing parallelogram.

    We can finally ensure that the maximum vertices number of feasible regions is 6, if the region of the slab covers any two non-adjacent vertices of two slab intersections, which is a parallelogram, then the slab will intersect all four edges of a parallelogram given four maximum intersection point. These four intersection points and two nonadjacent vertices of a parallelogram are vertices of the feasible region, which is 6 vertices.

    We mention that we do not focus on the non-degeneracy case of triangulation, since the case gives the result of the slab that lies inside another slab. 
\end{proof}

\begin{figure}[h!t]
	\centering
	\includegraphics[scale=0.45]{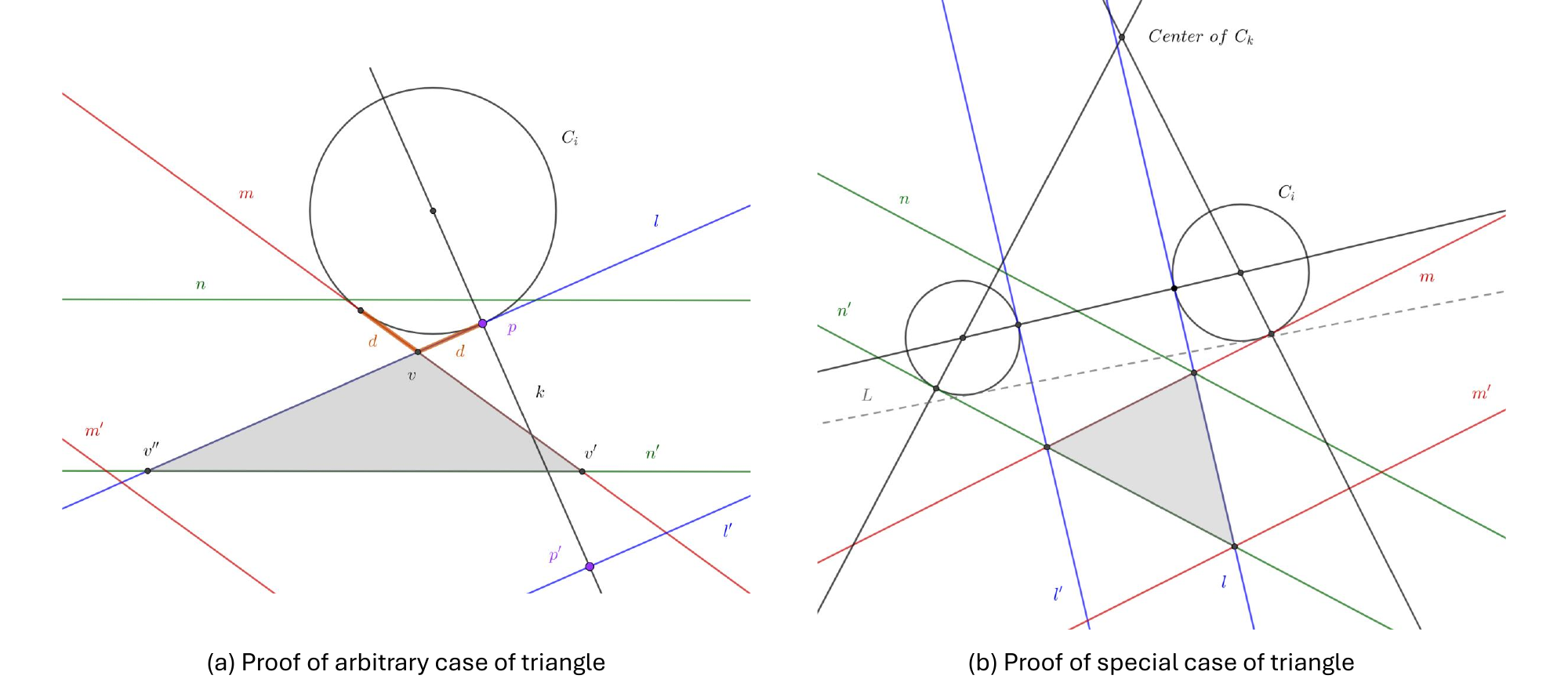}
	\caption{(a) proof of arbitrary case for the feasible region in triangular shape and (b) proof of special case that line of slab concurrent to be triangle for the feasible region in triangular shape.}
	\label{shape}
\end{figure}

\newpage

From our observation, if we consider a Voronoi diagram of a set of nonoverlapping circulars of the same radius, then the edges of the Voronoi diagram are the middle lines of slabs. Also, the vertices of Voronoi diagram are always in a feasible region.

\subsection{Algorithm to illuminate circular discs}

In construction of feasible regions on a set of circular discs, it is possible for feasible region overlapping. Since the feasible region is always a convex polygon, the intersection of feasible region still be a convex polygon. Therefore, we come up with some corollary that supports our algorithm in a further section of the study. 

\begin{corollary}
	Any point on the intersection of overlapping feasible regions illuminates all objective arcs of overlapping feasible regions.
\end{corollary}

\begin{proof}
	From \Tref{feasible}, any point on the feasible region will always illuminate their objective arcs. Therefore, intersection of overlapping feasible regions gives smaller regions that contain all points that still illuminate to their objective arcs.
\end{proof}

In our algorithm, we are going to use the fact of visibility on overlapping feasible regions to return a smaller number of a position set of illuminating points within the convex hull of a set of circular discs, which is called a curvilinear convex hull. Since any convex body can be illuminated by three illuminating points, in future work we aim to study the positioning of these three illuminating points outside the curvilinear convex hull of a set of circular discs. Therefore, in this study, we focus on positioning the illuminating point only inside the curvilinear.
\newpage
\subsubsection{Algorithm to positioning a set of illumination points}
\textbf{Input:} A set of $n$ circular center coordinates $C\left(x, y\right)$ and radius $r$.\\
\textbf{Output:} A set of feasible regions $Q$ with their centroid coordinates.
\begin{algorithm}
\caption{Feasible region computation}
\begin{algorithmic}[1]
	\State{Construct Laguerre Delaunay triangulation $C^*$ through 3D convex hull.}
	\State{Construct a set of intersection points $I$ between each edge of triangulation and boundaries of circular discs.}
	\State{Construct a set of tangent lines $L$ of each triangulation from the intersection points to circular discs.}
	\State{Construct a set of feasible region $F$ of each triangulation by halfplane intersection.}
	\State{Construct an empty set $Q$ and $R$.}
	\State{Construct a set of neighboring triangles $P$ of a set of feasible regions $F$}
	\State{Choose an element $P_i \in P$ with largest number of member.}
	\State{If the intersection region $J$ of members of $P_i$ is nonempty, then $J\in Q$ and remove members of $P_i$ from $F$.}
	\State{If $F\neq \emptyset$, then repeat from step 8 to step 11.}
	\State{Compute centroid of each region in set $Q$ and store in set $R$.}
	\State{Output $Q$ and $R$.}
\end{algorithmic}
\end{algorithm}

The complexity analysis of the algorithm involves several steps with varying computational costs. The construction of the Laguerre Delaunay triangulation with $n$ points has a complexity of $O(n\log n)$. The construction of intersection points between triangles and circular arc boundaries takes $O(n)$. The construction of slabs tangent to all intersection points requires constant time, $O(1)$, since the operation involves a fixed number of tangent lines. Determining feasible regions as the intersection of six half-planes has a complexity of $O(n)$. Lastly, constructing the set of neighboring triangles using a search algorithm also requires $O(n)$ as it involves a search through the triangles. The most computationally expensive operation in this process is the construction of the Laguerre Delaunay triangulation $O(n\log n)$, making the overall complexity of the process $O(n\log n)$.

\section{Experiment}

\subsection{Setting of Experiment}

We applied our algorithm with practical situation by generated position data on grid and close to grid point. Also, we varies the different diameter of circles, then calculate the feasible region through our procedure. Each tree will be represented as a cylindrical object, specifically the PVC pipe (15.24 cm) and the paper tube (7.62 cm), and will be placed at points in the grid or at a random position inside a circular disc with a radius half the size of the grid around the point in the grid (Table 1). 

Our algorithm will then be used to determine camera location regions within a given forest. The centroid of feasible regions from the output will be used to position cameras to illuminate our practical situation experiments. The study will include five experiments with specific settings described in the table below.

\newpage
\begin{longtblr}[
	caption = {The setting of each experiment, including numbers of materials used for tree representation and the position of the tree center.},
	]{
		width = \linewidth,
		colspec = {Q[110]Q[306]Q[113]Q[213]Q[194]},
		cells = {c},
		cell{4}{1} = {r=2}{},
		cell{4}{4} = {r=2}{},
		cell{4}{5} = {r=2}{},
		cell{6}{1} = {r=2}{},
		cell{6}{4} = {r=2}{},
		cell{6}{5} = {r=2}{},
		cell{8}{1} = {r=2}{},
		cell{8}{4} = {r=2}{},
		cell{8}{5} = {r=2}{},
		vlines,
		hline{1-4,6,8,10} = {-}{},
		hline{5,7,9} = {2-3}{},
	}
	Experiment No. & Material of trees representation with the diameters & Number of trees & Position of trees center          & Type of scanning position      \\
	1              & PVC pipes (15.24 cm)                                & 12              & On grid                           & Single scan and multiple scans \\
	2              & PVC pipes (15.24 cm)                                & 12              & Random position around grid point & Single scan and multiple scans \\
	3              & PVC pipes (15.24 cm)                                & 6               & On grid                           & Multiple scans                 \\
	& Paper tubes (7.62 cm)                               & 6               &                                   &                                \\
	4              & PVC pipes (15.24 cm)                                & 6               & Random position around grid point & Multiple scans                 \\
	& Paper tubes (7.62 cm)                               & 6               &                                   &                                \\
	5              & PVC pipes (15.24 cm)                                & 6               & Random position around grid point & Multiple scans                 \\
	& Paper tubes (7.62 cm)                               & 6               &                                   &                                
\end{longtblr}

\begin{figure}[h!t]
	\centering
	\includegraphics[scale=0.45]{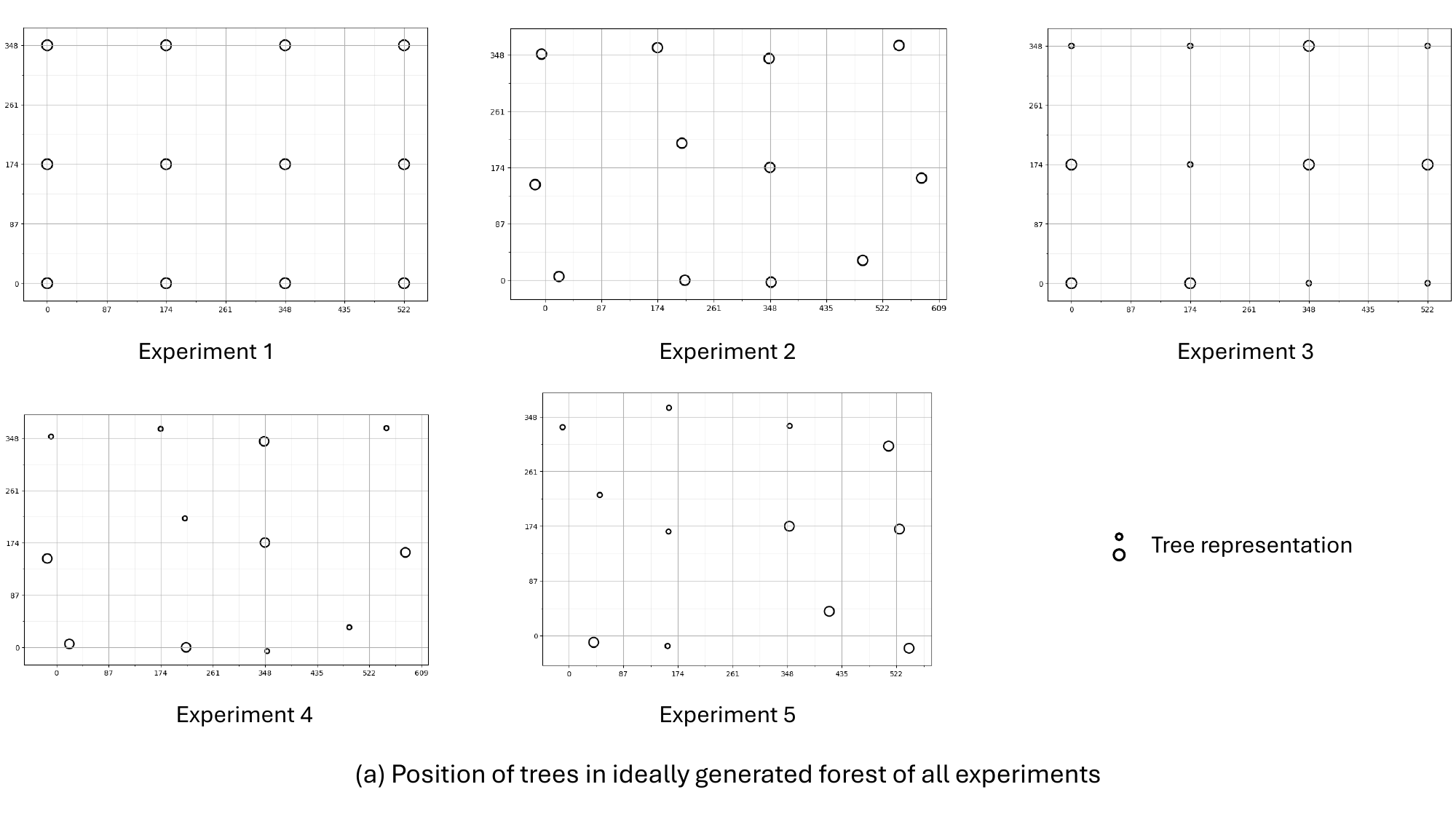}
\end{figure}

\begin{figure}[h!t]
	\centering
	\includegraphics[scale=0.45]{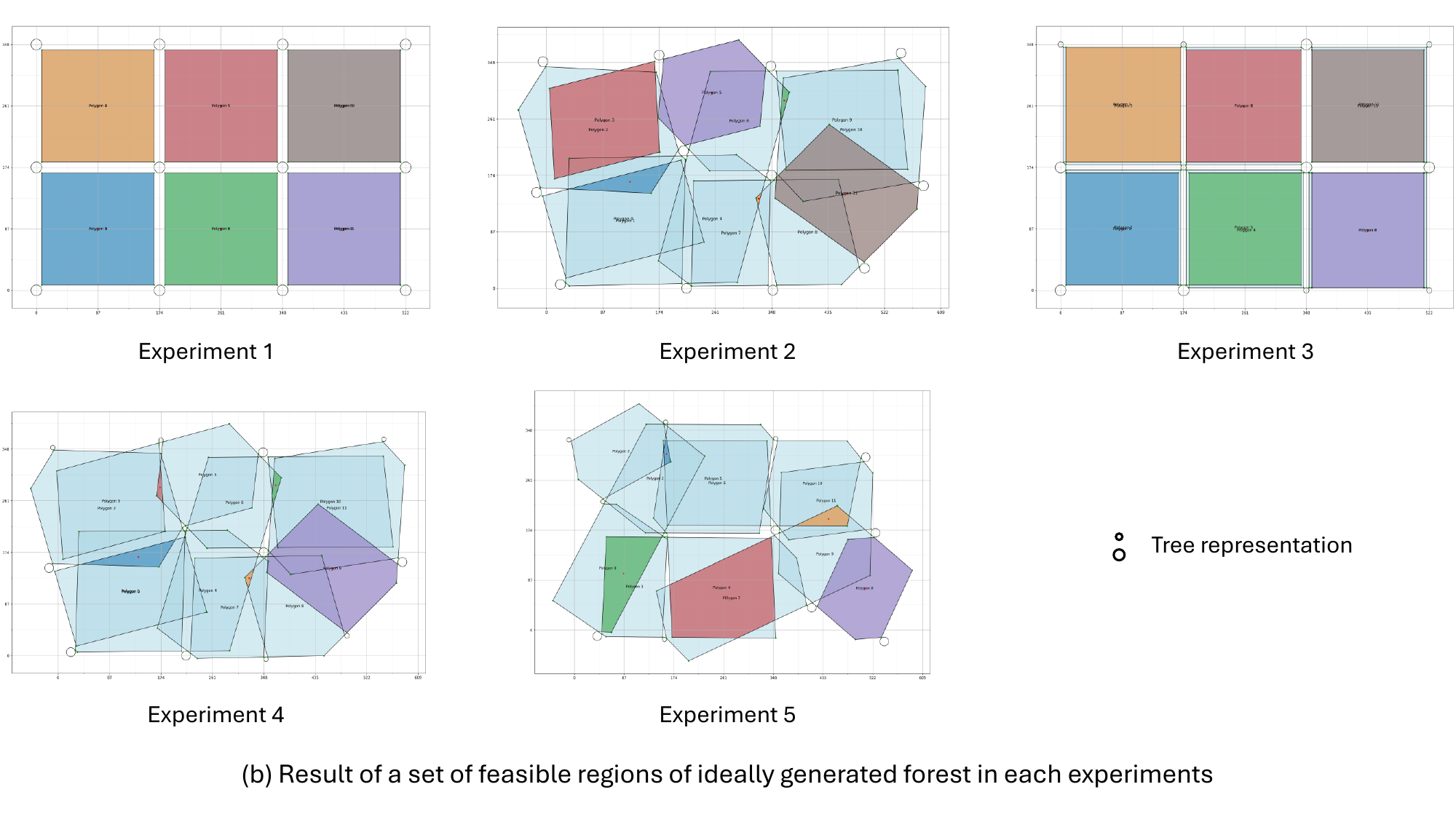}
	\caption{(a) position of trees in ideally generated forest of all experiments and (b) result of a set of feasible regions of ideally generated forest in each experiments}
	\label{dataresult}
\end{figure}

\newpage

\subsection{Experimental results and discussion}

The two experimental results include the number of detected trees from the TLS scannings showing the average of illuminated boundaries of a set of circular discs in percentage, and the theoretical average of illuminated boundaries. The latter use the 3DFIN plugin on CloudCompare to detect the trees diameter and calculate the illuminated boundary of each detected tree in the percentage. The former, theoretical results, are used in GeoGebra to construct illuminated boundaries as geometrical objects from generated data to collect the object values and then calculate in the percentage. The percentage of illuminated boundaries was determined with the standard deviration (SD) in Tables 2 and 3.

\subsubsection{Result on multiple scanning position}

\begin{longtblr}
	[
	caption = {Results from the TLS scannings and the theoretical generated the forest structure in each experiments.},
	]{
		width = \linewidth,
		colspec = {Q[117]Q[142]Q[346]Q[333]},
		cells = {c},
		hlines,
		vlines,
	}
	Experiment No. & No. of detected tree & Experimental average of illuminated boundaries (\%, ±SD) & Theoretical   average of illuminated boundaries (\%, ±SD) \\
	1              & 12                   & 77.26 ±   13.21                                     & 77.83 ± 12.16                                        \\
	2              & 12                   & 72.51 ± 14.92                                       & 73.42 ± 14.01                                        \\
	3              & 7                    & 78.37 ± 11.41                                       & 82.05 ± 13.31                                        \\
	4              & 11                   & 71.28 ± 13.82                                       & 72.66 ± 14.77                                        \\
	5              & 7                    & 77.18 ± 14.61                                       & 81.47 ± 15.30                                        
\end{longtblr}

\subsubsection{Result on single scanning position}

\begin{longtblr}[
	caption = {Result of experiment number 1 and experiment number 2 in the single scan.},
	]{
		width = \linewidth,
		colspec = {Q[108]Q[75]Q[131]Q[315]Q[306]},
		cells = {c},
		cell{2}{1} = {r=6}{},
		cell{8}{1} = {r=6}{},
		vlines,
		hline{1-2,8,14} = {-}{},
		hline{3-7,9-13} = {2-5}{},
	}
	Experiment No. & Scan No. & No. of detected tree & Experimental average of illuminated boundaries (\%, ±SD) & Theoretical average of illuminated boundaries (\%, ±SD) \\
	1              & 1        & 11                   & 47.60 ±
	2.70                                      & 48.82 ± 0.65                                       \\
	& 2        & 10                   & 49.10 ± 1.67                                        & 48.68 ± 0.57                                       \\
	& 3        & 11                   & 47.54 ± 3.13                                        & 48.82 ± 0.65                                       \\
	& 4        & 11                   & 47.98 ± 2.72                                        & 48.82 ± 0.65                                       \\
	& 5        & 10                   & 49.44 ± 1.38                                        & 48.68 ± 0.57                                       \\
	& 6        & 11                   & 47.98 ± 2.77                                        & 48.82 ± 0.65                                       \\
	2              & 1        & 11                   & 47.73 ±
	2.93                                      & 48.91 ± 0.58                                       \\
	& 2        & 10                   & 49.17 ± 2.17                                        & 48.64 ± 0.77                                       \\
	& 3        & 11                   & 47.92 ± 2.64                                        & 48.86 ± 0.64                                       \\
	& 4        & 11                   & 47.16 ± 3.05                                        & 48.81 ± 0.55                                       \\
	& 5        & 11                   & 46.72 ± 3.79                                        & 48.43 ± 1.49                                       \\
	& 6        & 12                   & 47.16 ± 2.45                                        & 48.72 ± 1.03                                       
\end{longtblr}

\subsubsection{Experimental discussion}

Compatibility between TLS scanning and the region derived from the proposed methods for tree stem detection at DBH, as a cylindrical object, can arise when the position of the trees centers are on the grid, which shows 78. 37\% and 82. 05\% of the bounces of experiment number 3 and the theory, respectively. However, the small represented trees (diameter 15.24 cm) were only 58.33\% detectable from multiple scans.

From single scan experiments number 1 and 2, we derived the lower illuminated boundaries from about 47\% to 49\% that showed the proper accuracy ranges, less than maximum 50\%, since one illuminating point cannot illuminate half of the boundary of the convex body \cite{illuprob}, for tree stem detection. In addition, scan number 6 in experiment number 2 gave the 100\% detected trees from a single scan (Table 3 and Figure \ref{allresult}).

\newpage

\begin{figure}[h!t]
	\centering
	\includegraphics[scale=0.42]{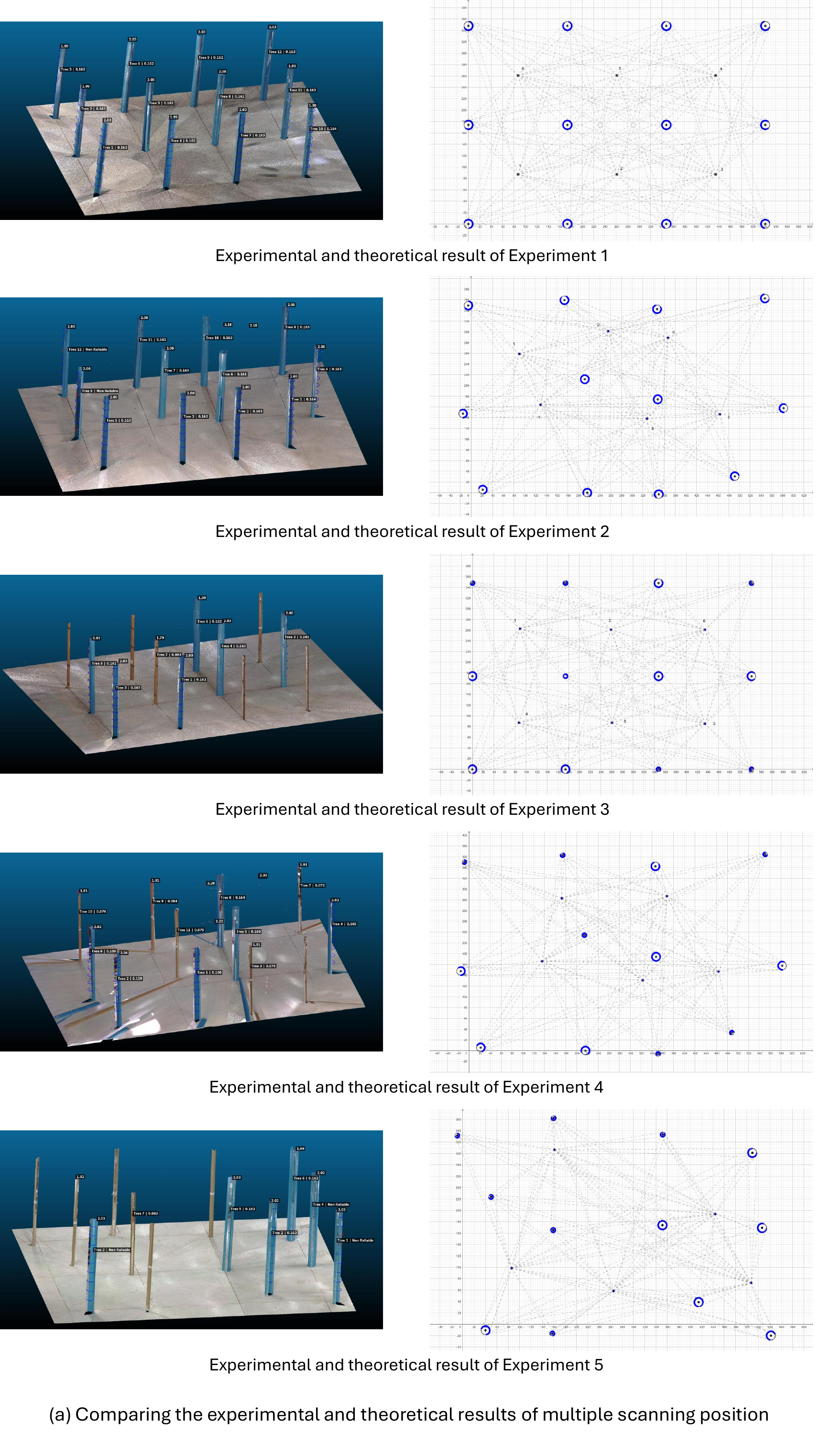}
\end{figure}

\begin{figure}[h!t]
	\centering
	\includegraphics[scale=0.43]{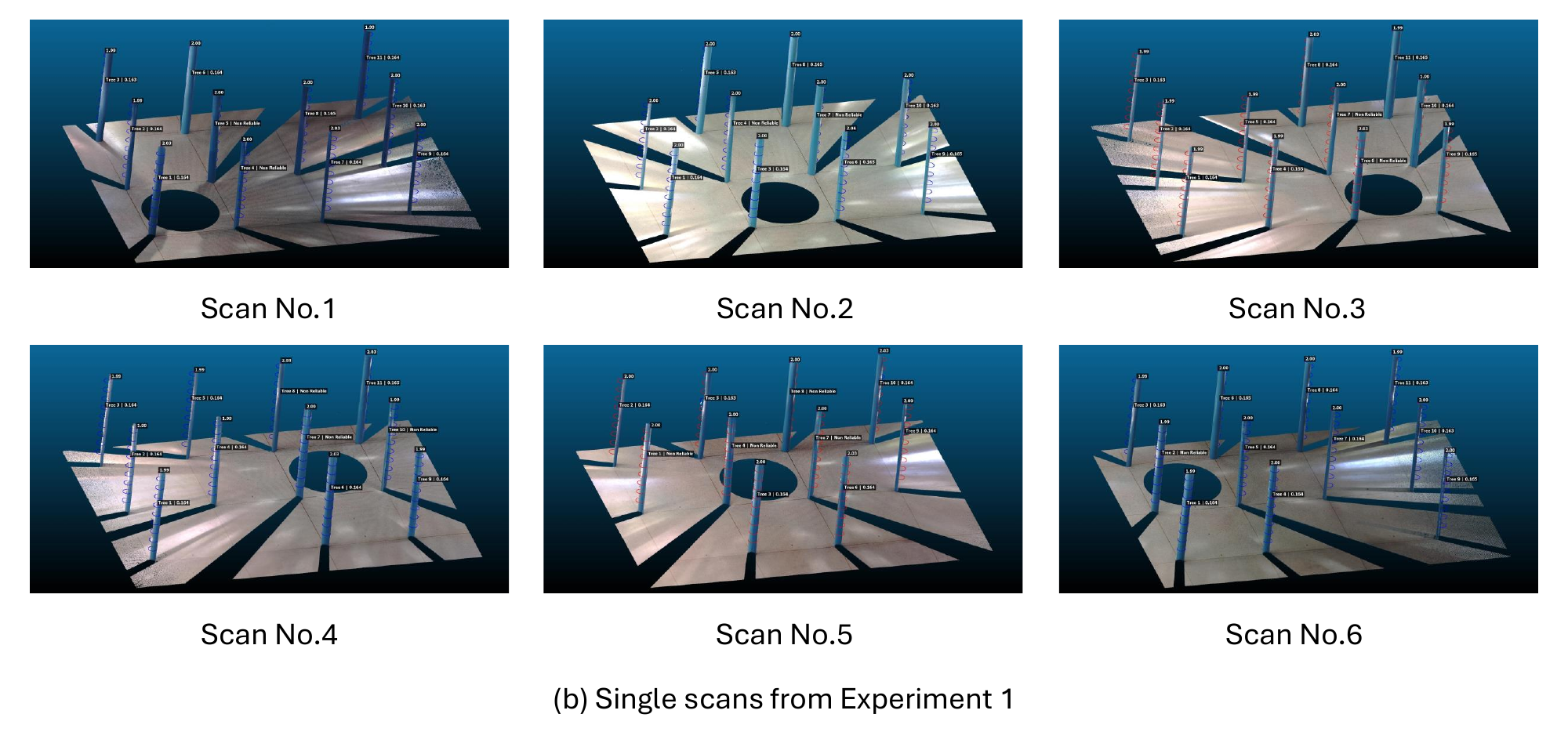}
	\includegraphics[scale=0.43]{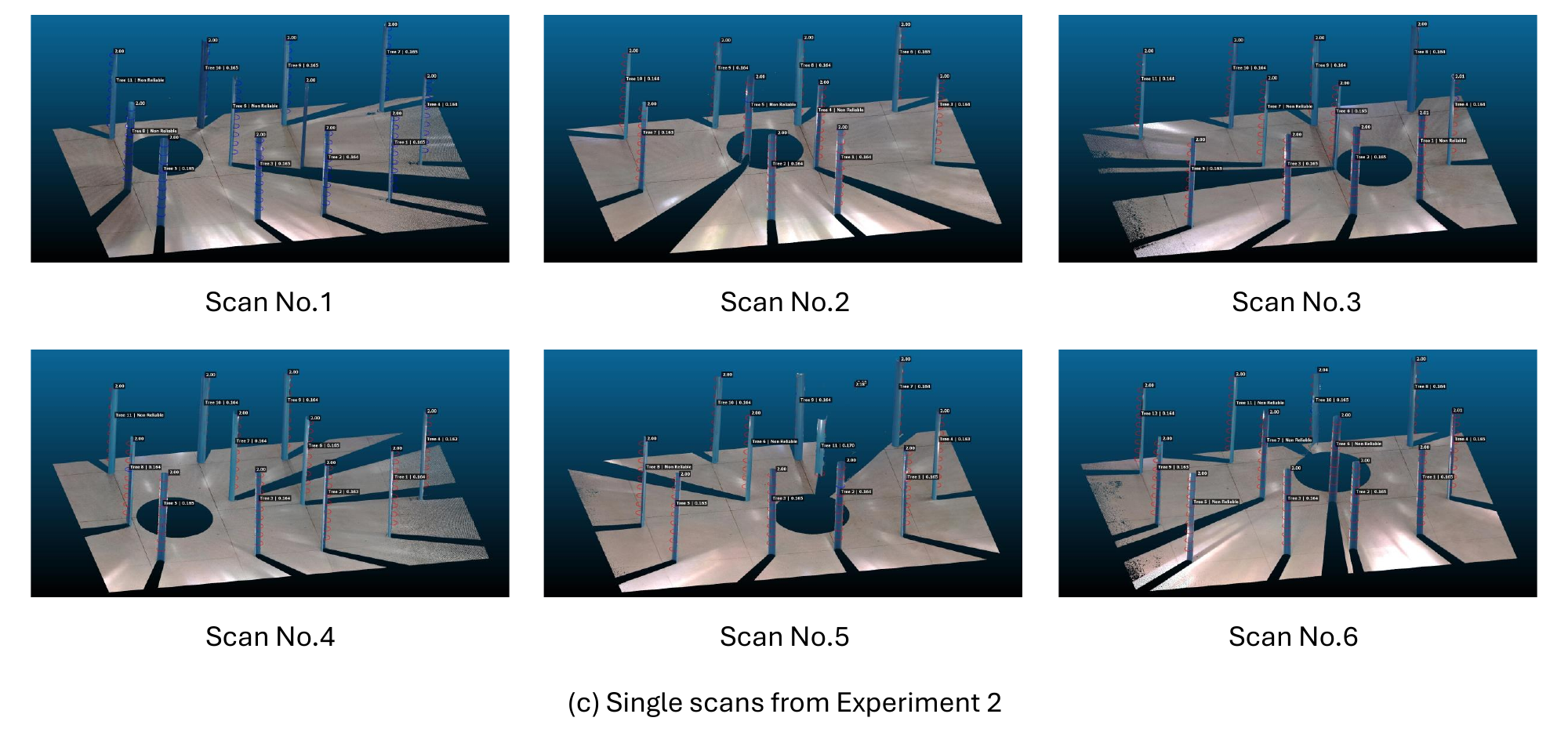}
	\caption{(a) comparing the experimental and theoretical results of all experiments, the full and partial blue cross-sectional circles represent the illuminated bounaries of trees, (b) single scans from experiment number 1 and (c) single scans from experiment number 2.}
	\label{allresult}
\end{figure}

\section{Concluding remark}

We proposed a geometrical algorithm to find a set of feasible regions to position the illuminating point given as LiDAR and camera sensors to illuminate a set of circular discs given as a set of tree stems from the cross section. The algorithm uses geometric methods such as the Laguerre Delaunay triangulation and the slab to specify feasible regions for positioning the illuminating point and studying their shape. Also, we define an objective arc, which is constructed through the triangulation edges and circular disc boundaries, to be the least arc of a circular disc boundary illuminated by an illuminating point. From these geometrical objects, we come up with some theory to support the algorithm.

This study applied the algorithm to five experiments with practical situations by generating position data on the grid and close to the grid point and varying the different diameters of circles as tree stems in Table 1. Moreover, we provide two different results of the scanning position, which are multiple scan and single scan. Both different scanning positions give compatible results between TLS scanning and the region derived from the proposed methods for tree stem detection at DBH in Table 2 and Table 3. The result on average of illuminated boundaries in the percentage of a single scanning position gives proper accuracy ranges for tree stem detection.

For future work, we plan to prove mathematical methods to minimize the number of scans in tree stem illumination problems both inside and outside the curvilinear of a set of circular discs. Also, we tend to find possible shapes of the feasible region to give patterns of triangulation alignment for positioning the illuminating point in practical situations. Furthermore, we aim to improve the geometrical algorithm to perform the experiment with a real forest.

\acknowledgements{The first author would thank the  master's degree friends from the Department of Mathematics, Faculty of Science, Chiang Mai University, especially Asama Jampeepan for great discussion while working on this paper. The discussions were done through Geometry Boot Camp 2023-2024, especially Nitipol Nitipon Moonwichit and Nattawut Phetmak. The first author is supported by the Teaching Assistant and Research Assistant Scholarship for the Academic Year 2023-2024 from Chiang Mai University. This research is partially supported by Chiang Mai University, and the experiment in this research is supported by FORRU's Small-Grant Support for Student Projects.}


\end{document}